\documentclass[conference]{IEEEtran}

\usepackage[T1]{fontenc}
\usepackage{graphicx}
\usepackage{amsmath,amssymb}
\usepackage{booktabs}
\usepackage{microtype}
\usepackage{hyperref}
\hypersetup{hypertexnames=false,hidelinks}
\usepackage{cleveref}
\usepackage{tikz}
\usetikzlibrary{arrows.meta,positioning,calc,decorations.pathreplacing}
\usepackage{xspace}
\usepackage{multirow}
\usepackage{array}
\usepackage{pifont}
\usepackage{algorithm}
\usepackage{algorithmic}
\usepackage{cite}
\usepackage{amsthm}

\usepackage{textcomp}
\DeclareFontShape{T1}{ptm}{m}{scit}{<->ssub*ptm/m/sc}{}

\newtheorem{theorem}{Theorem}
\newtheorem{definition}[theorem]{Definition}

\newcommand{\linkedlist}{\textsc{Linked-List}\xspace}
\newcommand{\merkle}{\textsc{Merkle}\xspace}
\newcommand{\signeddag}{\textsc{Signed-DAG}\xspace}
\newcommand{\nochain}{\textsc{No-Chain}\xspace}
\newcommand{\hashonly}{\textsc{Hash-Only}\xspace}
\newcommand{\naivesig}{\textsc{Naive-Sig}\xspace}

\newcommand{\Bone}{B1\xspace}
\newcommand{\Btwo}{B2\xspace}
\newcommand{\Bthree}{B3\xspace}
\newcommand{\Bfour}{B4\xspace}
\newcommand{\Bfive}{B5\xspace}

\begin{document}

\title{Attesting Outputs and Delegation Ancestry in Multi-Agent AI Systems}

\author{\IEEEauthorblockN{Lifei Liu, Haoran Yu}
\IEEEauthorblockA{\textit{Independent Researcher}\\
Seattle, WA, USA\\
\texttt{lliu.lifei@gmail.com}, \texttt{haoranyu889@gmail.com}\\
Corresponding author: Lifei Liu}}

\maketitle

\begin{abstract}
Multi-agent applications delegate work across independently operated
deployers. After an incident, a verifier must answer two questions: which
deployer released the reported bytes, and whether each cross-deployer edge was
authorized. Credentials establish who may act, but need not bind them to later
output bytes or prove both deployers authorized a dynamically created edge.

We present a two-layer attestation design for dynamic delegation without a
shared authority, public log, or precommitted workflow. A trusted deployer
runtime signs a hash of each released output; this records released bytes but
does not prevent prompt injection. Ancestry evidence records edge
authorization. Under a unified threat model, we compare a signed linked list,
a Merkle-chain variant, and a co-signed DAG. The primitives are standard; the
contribution is deployer-side binding and the evidence needed for the two
questions. After child-key compromise, the single-signer designs permit an
unauthorized parent binding, whereas the co-signed DAG rejects it because the
parent must authorize the edge. Fixed adversary matrices and regression tests
validate the composed verifier. On an Apple M1 Pro, ancestry-only checks take
24.3--499.2\,$\mu$s per hop. In a live local multi-service workflow, a parent
discovers the child's A2A Agent Card; the child calls an MCP tool and releases
local-LLM output: all
30 signed-DAG tasks passed complete verification, while a controlled
child-key-only claim was rejected; its mean end-to-end latency was 813.1\,ms
versus 770.8\,ms without evidence. In a complementary three-availability-zone
AWS deployment, all 1{,}000 valid co-signed-DAG paths verified; issuance
averaged 3.651\,ms and complete verification 5.015\,ms. The cloud result
excludes TLS/mTLS, KMS, and model-serving latency.

\end{abstract}


\section{Introduction}
\label{sec:introduction}

AI agents increasingly complete tasks through multi-hop delegation across
deployment boundaries~\cite{yao2023react,wu2023autogen}.
Consider a deployer that delegates a task to an independently operated agent.
After that agent releases an output, incident response must answer two
questions: which deployer released these exact bytes, and did the parent
deployer authorize this particular cross-deployer edge? Authorization
credentials record a principal and permission at issuance time; they need not
bind later output to a trusted runtime or establish a complete, authorized path.

These are post-incident evidence questions, not claims about agent behavior.
Indirect prompt injection can alter an agent's behavior after an authorization
check has succeeded~\cite{greshake2023indirect}. A signature cannot establish
that an output was benign or prevent its production, but it can provide
tamper-evident evidence of the bytes emitted and the delegation path for
investigation and targeted revocation.

Existing delegation mechanisms identify an actor or grant a
permission~\cite{south2025delegation,a2a2026spec}, but, in this setting, leave
two evidence gaps for dynamically created paths. \emph{G1 (output binding):} a
credential does not attest bytes later released by a runtime. \emph{G2 (edge
consent):} under child-key compromise, a child-only signature on a new parent
reference does not demonstrate that the named parent approved the relationship.

Our key insight is that these gaps require two independently verifiable
evidence layers rather than a stronger credential. A trusted deployer runtime
signs a hash of each released output, while an ancestry record carries the
authorization evidence for each edge. The verifier must decide from a presented
path because the setting may lack a shared authorization service, public log, or
precommitted workflow. The primitives are standard; our contribution is the
deployer-side binding and the evidence required for this dynamic setting. We
compare a signed linked list, a Merkle-chain variant, and a co-signed DAG. After
child-key compromise, only the co-signed DAG rejects a new parent claim because
it requires the parent's signature on the edge. Figure~\ref{fig:two_layer_architecture}
illustrates the two evidence layers and this distinction.

\begin{figure*}[t]
\centering
\resizebox{0.75\textwidth}{!}{%
\begin{tikzpicture}[
  x=1cm, y=1cm, >=Latex, font=\scriptsize,
  actor/.style={draw=blue!65!black, line width=.55pt, rounded corners=2pt,
    fill=blue!4, minimum width=1.22cm, minimum height=.46cm,
    align=center, font=\scriptsize\bfseries},
  evidence/.style={draw=green!45!black, line width=.55pt, rounded corners=2pt,
    fill=green!4, minimum width=1.28cm, minimum height=.46cm, align=center},
  record/.style={draw=violet!75!black, line width=.55pt, rounded corners=2pt,
    fill=violet!4, minimum width=1.24cm, minimum height=.46cm, align=center},
  verdict/.style={draw=orange!85!black, line width=.55pt, rounded corners=2pt,
    fill=orange!4, minimum width=1.05cm, minimum height=.46cm, align=center},
  attack/.style={draw=red!75!black, line width=.6pt, rounded corners=2pt,
    fill=red!4, minimum width=1.76cm, minimum height=.46cm, align=center},
  flow/.style={->, line width=.6pt, draw=black!60},
  outputflow/.style={->, line width=.7pt, draw=green!45!black},
  ancestryflow/.style={->, line width=.7pt, draw=violet!75!black}
]
\node[anchor=west, font=\scriptsize\bfseries] at (.05,3.74) {Normal delegation};
\node[actor] (agentA) at (1.05,3.20) {Agent A};
\node[actor] (agentB) at (3.25,3.20) {Agent B};
\draw[flow] (agentA) -- node[above, font=\scriptsize, text=black!60] {delegates} (agentB);

\node[anchor=west, font=\scriptsize, text=green!40!black] at (.05,2.63)
  {released bytes};
\node[evidence, minimum width=1.48cm] (runtime) at (5.05,2.45) {$D_B$ runtime};
\node[evidence, minimum width=1.37cm] (receipt) at (7.05,2.45)
  {$H(o_B),\;\sigma_B$};
\node[verdict, minimum width=1.16cm] (outcheck) at (8.70,2.45)
  {attribute};
\draw[outputflow] (agentB.south east) -- node[pos=.38, above=2pt,
  font=\scriptsize, text=green!40!black] {$o_B$} (runtime.north west);
\draw[outputflow] (runtime.east) -- (receipt.west);
\draw[outputflow] (receipt.east) -- (outcheck.west);

\node[anchor=west, font=\scriptsize, text=violet!75!black] at (.05,2.00)
  {delegation path};
\node[record] (root) at (1.50,1.70) {$r_A$: root};
\node[record, minimum width=2.08cm] (child) at (4.55,1.70)
  {$r_B$: parent $H(r_A)$};
\node[verdict, minimum width=1.16cm] (edgecheck) at (8.70,1.70)
  {authorize};
\draw[ancestryflow] (root) -- (child);
\draw[ancestryflow] (child) -- (edgecheck);

\node[anchor=west, font=\scriptsize\bfseries, text=red!65!black]
  at (.05,1.30) {After compromise of $K_B$};
\node[attack] (forged) at (1.38,.67) {forged $r'_B$: parent $r_A$};
\node[record, minimum width=1.55cm] (single) at (4.70,.92)
  {child-only $\sigma_B$};
\node[attack, minimum width=1.00cm] (accept) at (7.35,.92) {ACCEPT};
\draw[red!75!black, dashed, ->, line width=.7pt] (forged.east) |- (single.west);
\draw[red!75!black, dashed, ->, line width=.7pt] (single) -- (accept);

\node[record, minimum width=1.55cm] (cosigned) at (4.70,.28)
  {\signeddag: $\sigma_A+\sigma_B$};
\node[verdict, minimum width=1.24cm] (reject) at (7.35,.28) {REJECT};
\draw[ancestryflow] (forged.east) |- (cosigned.west);
\draw[ancestryflow] (cosigned) -- (reject);
\node[anchor=west, font=\scriptsize, text=violet!75!black] at (8.18,.28)
  {$\sigma_A$ absent};
\end{tikzpicture}
}
\caption{Motivating example: output receipts and delegation records answer
different forensic questions. The deployer runtime binds a receipt to the
released bytes; after compromise of $D_B$'s key, a child-only chain accepts a
forged parent claim, whereas \signeddag requires both endpoint signatures and
rejects it when $\sigma_A$ is absent.}
\label{fig:two_layer_architecture}
\end{figure*}
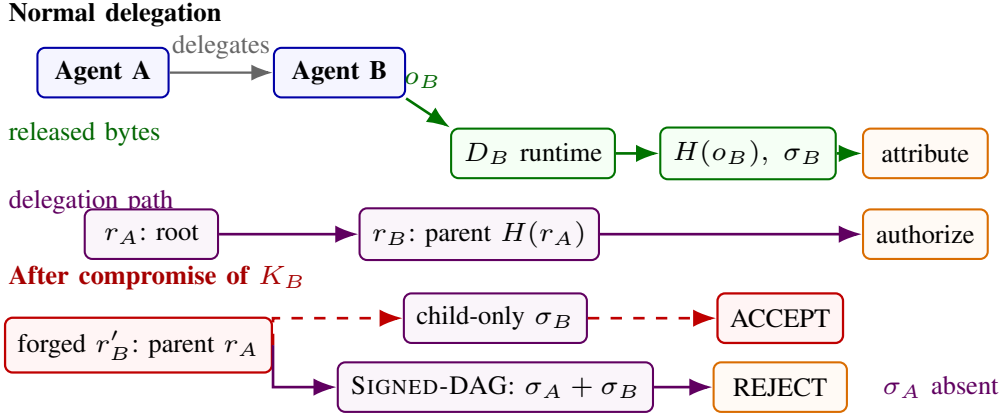

We make the following contributions:

\begin{enumerate}
\item \textbf{A deployer-side output-attestation interface.} A deployer runtime,
rather than the LLM agent, hashes each attested output and signs the complete
record with deployer and validity metadata. The interface records what was
released under a trusted-runtime assumption; it does not detect or prevent
prompt injection.

\item \textbf{A unified analysis of ancestry evidence.} We formulate a single
system and threat model for both layers, then evaluate three standard
ancestry-binding constructions as comparative designs. The key result is that
a signed linked list and a Merkle-chain variant use only the child deployer's
authorization for an edge; a co-signed DAG additionally requires the parent
deployer's consent and thus resists unauthorized cross-deployer edge creation
after child-key compromise.

\item \textbf{Implementation and deployment evidence.} We validate the specified
predicates with fixed matrices and regression tests. We measure issuance and
checking cost in LangGraph, process-separated, container, and
three-availability-zone AWS deployments, and execute a live local multi-service
workflow with A2A JSON-RPC, an MCP Streamable-HTTP tool, and LLM-produced child
output. Historical LLM-generated aggregates are exploratory context, not
current security evidence.
\end{enumerate}

The result is a bounded deployment guideline: a linked representation is
appropriate for a single trust domain, while resistance to unauthorized
cross-deployer edges after child-key compromise requires explicit parent
authorization. We do not claim that the construction prevents agent compromise;
it makes the resulting outputs and delegation path auditable.

\section{Background and Related Work}
\label{sec:related_work}

\paragraph{Agent authorization and behavioral evidence}
Agent-identity work distinguishes authorization from the behavior that follows
it. Indirect prompt injection exploits this separation~\cite{greshake2023indirect},
and high probe AUC alone need not establish malicious-content
detection~\cite{li2026auc}. Chan et al.~\cite{chan2024visibility} identify
behavioral verification as a missing capability. MCP standardizes tool and context exchange, while A2A
defines protocol-level authentication and authorization requirements for
inter-agent communication~\cite{mcp2025spec,a2a2026spec}. Our deployer-side
binding records released bytes; it neither judges semantic safety nor prevents
their production.

\paragraph{Agent ancestry and delegation protocols}
W3C DID and Verifiable Credentials~\cite{w3c2022did,w3c2022vc} provide
single-hop primitives, while the IETF agent-auditing architecture records
delegation transitions in a broader interoperable framework~\cite{kuehlewind2026auditarchitecture}.
South~\cite{south2025delegation} addresses authenticated delegation of
authority to agents. Llamb\'{i}-Morillas and Fern\'{a}ndez-Fern\'{a}ndez~\cite{llambi2026cva}
instead formalize evidence that a concrete request satisfies a policy in an
execution context. Our narrower comparison binds released output bytes and
identifies the signatures needed for a cross-deployer edge after child-key
compromise; it can supply evidence to these authorization and audit layers.

\paragraph{General provenance and chain integrity}
in-toto~\cite{torresarias2019intoto} validates a precommitted software
supply-chain layout, while SPIFFE/SPIRE issues workload identities from
registered trust domains and supports their update propagation~\cite{spiffe2026workloadapi}.
Our setting instead creates
delegations on demand across independently operated deployers, without a public
log or predeclared graph.

\paragraph{Concurrent agent-identity protocols}
AIP~\cite{prakash2026aip} uses append-only capability tokens and an optional
agent-signed final-result hash. HDP~\cite{dalugoda2026hdp} excludes key
compromise, MemLineage~\cite{ouyang2026memlineage} targets memory integrity,
and Mandato~\cite{racioppi2026mandato} targets pre-action policy conformance.
Our target is post-incident evidence of released bytes and bilateral deployer
authorization of an edge. Table~\ref{tab:comparison} compares the concrete
identity and ancestry proposals.

\begin{table}[t]
\centering
\caption{Comparison with concurrent agent-identity systems. Here,
``deployer-side'' means that a signer separate from the LLM/agent observes
released segments; AIP's optional agent-signed final result hash is not counted.}
\label{tab:comparison}
\footnotesize
\begin{tabular}{@{}lcccc@{}}
\toprule
Property & AIP & HDP & MemLineage & Ours \\
\midrule
Deployer-side output binding & \ding{55} & \ding{55} & \ding{55} & \checkmark \\
Multi-hop chain integrity & \checkmark & \checkmark & \checkmark & \checkmark \\
Comparative chain analysis & \ding{55} & \ding{55} & \ding{55} & \checkmark \\
Parent edge co-signature & \ding{55} & \ding{55} & \ding{55} & \checkmark \\
Child-key-compromise analysis & \ding{55} & Excluded & \ding{55} & \checkmark \\
Structural-verifier testing & \ding{55} & \ding{55} & \ding{55} & \checkmark \\
\bottomrule
\end{tabular}
\end{table}

\paragraph{Complementary controls}
Information-flow control~\cite{myers1997dlm} addresses behavioral compromise.
JWT~\cite{jones2015jwt} and OAuth~\cite{hardt2012oauth} express signed claims
and delegated permissions, but do not by themselves bind later bytes to a
credential or establish a complete ancestry path.

\section{System and Unified Threat Model}
\label{sec:threat_model}

\subsection{System Model and Security Objectives}

We consider a principal that delegates a task through agents operated by
deployers $D_A,D_B,\ldots$. A deployer is the service that hosts an agent,
controls its signing key, and emits a delegation record when the agent invokes
another agent. Each record identifies the issuer, child instance, proposed
parent, declared depth, expiry time, nonce, and evidence construction. A
consumer receives the records and verifies a claimed path from a root
principal to a leaf agent. The prototype has no separately named global
\texttt{chain\_id}; a chain is the ordered, presented root-to-leaf path.

The design produces two kinds of evidence. \emph{Output evidence} binds a
deployer-issued credential to bytes released by its runtime. \emph{Ancestry
evidence} binds each delegation record to a parent and, for cross-deployer
edges, can require both endpoint deployers to approve the relationship. The
objectives are therefore: (O1) reject substituted output bytes and replay at a
verifier with persistent nonce state; (O2) reject malformed, truncated, or
spliced ancestry presented to the
consumer; and (O3) under a single child-deployer key compromise, reject a
new cross-deployer edge not authorized by its parent. These are
post-incident evidence objectives. They do not prevent an LLM from following a
malicious instruction, assess whether a signed output is safe, or enforce
application-specific authorization policy.

\subsection{Unified Trust Boundary}

We trust four components: the deployer runtime to observe and sign the bytes
it releases, its private signing key, the registry that admits authorized
deployers, and the verifier implementation. Everything outside those
components, including the LLM agents, tool responses, network transport, and
supplied chain documents, may be malicious. This boundary applies to both layers:
Layer~1 relies on the deployer runtime for faithful observation, and Layer~2
relies on the registry and verifier to interpret a delegation record.

The adversary can control tool responses and agent behavior, observe public
keys and verifier source, replay or modify evidence in transit, and construct
arbitrary structural payloads. It cannot forge a signature under an
uncompromised key, find a collision or preimage that defeats the hash binding,
compromise the trusted runtime, or enroll an arbitrary deployer in the
admission-controlled registry. The single-key compromise study relaxes the
secrecy assumption for one deployer's signing key. Table~\ref{tab:trust}
states the consequence of each trust-boundary violation.

\begin{table}[t]
\centering
\caption{Unified trust assumptions and their consequences.}
\label{tab:trust}
\footnotesize
\setlength{\tabcolsep}{2pt}
\begin{tabular}{@{}p{1.7cm}p{2.35cm}p{3.05cm}@{}}
\toprule
Component & Assumption & Consequence if violated \\
\midrule
Agent / LLM & Untrusted & It may emit harmful output; attestation records it. \\
Deployer runtime & Faithfully signs released bytes & Output evidence no longer reflects runtime output. \\
Signing key & Secret, except in §\ref{sec:results_keycompromise} & A compromised child motivates parent co-signature; both endpoints defeat all designs. \\
Key registry & Admission-controlled & Self-enrolled identities become indistinguishable from legitimate deployers. \\
Verifier & Enforces stated predicates & Invalid evidence can be accepted; addressed by tests and red-team mutation. \\
Transport & Untrusted & Signatures and hashes detect modification. \\
\bottomrule
\end{tabular}
\end{table}

\subsection{Cryptographic Instantiation}

Ed25519 and SHA-256 are concrete implementation choices, not architectural
requirements. The analysis uses the abstract properties of strong signature
unforgeability under chosen-message attack and hash collision resistance (and,
where required, preimage resistance). An implementation may substitute
equivalent primitives provided it preserves an unambiguous encoding, domain
separation, and the same security properties. The prototype uses deterministic
Python JSON serialization (recursive key sorting, compact separators, and
ASCII escaping), Ed25519~\cite{bernstein2012ed25519}, and SHA-256. This
serialization is deterministic within the prototype but is \emph{not} a full
interoperable implementation of RFC~8785/JCS~\cite{rundgren2020jcs}.

\subsection{Structural Attacks and Scope}
\label{sec:taxonomy}

The adversary's structural goal is to make a verifier accept an ancestry that
does not represent the actual delegation. We evaluate four graph-edit classes:
orphan insertion (B1), middle-node removal and relinking (B2), transplanting a
node across chains (B3), and revising an ancestor after children exist (B4).
Sibling fanout (B5) is different: it creates several valid children under a
valid parent and therefore requires an application policy such as a signed
child limit plus verifier state. Table~\ref{tab:scope} distinguishes these
cases from threats outside the protocol boundary.

\begin{table}[t]
\centering
\caption{Threat scope and corresponding verifier state.}
\label{tab:scope}
\scriptsize
\setlength{\tabcolsep}{2pt}
\begin{tabular}{@{}p{1.6cm}p{1.05cm}p{3.35cm}p{1.0cm}@{}}
\toprule
Threat & Status & Mechanism & State \\
\midrule
B1--B4 structural edits & In scope & Signature and ancestry predicates & None \\
B5 sibling fanout & Policy-dependent & Signed limit and child counter & Per parent \\
Truncation & In scope & Signed ancestry depth & None \\
Replay & In scope & Nonce uniqueness and expiry & Per verifier \\
Output substitution & In scope & Signed output hash & Per verifier \\
Output order / sessions & Not implemented & Needs signed sequence fields & Per session \\
Canonical encoding & Prototype scope & Deterministic local serializer; not JCS & None \\
Key registry & Out of scope & Admission control is trusted & $-$ \\
Endpoint collusion & Out of scope & Both keys compromised & $-$ \\
Prompt-injection prevention & Out of scope & Requires IFC / monitoring & $-$ \\
\bottomrule
\end{tabular}
\end{table}

The taxonomy deliberately addresses evidence manipulation, not the
attacker's incentive after causing harm. A forged, shortened, or reattached
path may not undo an action, but it can misattribute the deployer responsible
for that action and obstruct revocation or incident investigation. We make no
claim that chain integrity alone deters every prompt-injection attacker.

\subsection{Holistic Security Rationale}

An output verifier accepts only when the signed output hash, primary
signature, key identifier, validity interval, and nonce predicate match the
bytes it received. An ancestry verifier accepts only when each record has a
registered issuer, its parent reference and declared depth are consistent, and
any required cross-deployer co-signatures are present. The two checks use
separate signed fields but share the trust boundary above: failure of an output
check does not satisfy an ancestry predicate, and a structurally valid chain
does not authenticate altered output bytes.

Consequently, under the primary model, an accepted tampered output would imply
either a failure of the signature/hash assumptions or a violation of the
trusted deployer runtime. An accepted structural edit would imply a missing
verifier predicate, a cryptographic break, or a compromised signing key. The
key-compromise analysis in Section~\ref{sec:results_keycompromise} isolates the
last case. Adaptive and lifecycle campaigns in Section~\ref{sec:results} test
whether the implementation actually enforces these stated predicates; they do
not empirically validate Ed25519 or SHA-256.

\subsection{Registry and Replay Semantics}\label{sec:trust_bootstrap}

The verifier maintains a local registry of authorized deployer public keys,
analogous to a TLS trust store. The prototype checks that each envelope
\texttt{key\_id} is derived from the pinned public key and can locally revoke
that identifier. It does not implement a production key-rotation protocol,
durable distributed revocation propagation, or a cross-verifier replay store.
A caller can attach a persistent nonce store to reject replay at one verifier instance. These
requirements do not introduce a global delegation authority, but they do mean
that the protocol does not protect an open ecosystem in which an attacker can
self-register a deployer. A signed-DAG edge binds the signed construction
identifier, the co-signer set, and the child's nonce in its edge manifest;
there is no separate edge-nonce field in the current schema.

\section{Output Attestation (Layer~1)}
\label{sec:output_binding}

Output attestation addresses a narrow question: given bytes received by a
consumer, can it verify that a particular deployer runtime released those
bytes? This is a general cryptographic operation, and we do not claim the
signature construction itself as an AI-specific primitive. Its relevance here
is placement: the deployer performs the binding after the runtime produces an
output and before the output is delivered or delegated. The agent neither
holds the signing key nor populates the signed fields.

\subsection{Signed Output Record}

For each attested output, the prototype emits a versioned record containing
the instance and system identifiers, deployer name, creation and expiry times,
parent reference, declared ancestry depth, chain proof, SHA-256
\texttt{output\_hash}, a 128-bit \texttt{nonce}, a signed
\texttt{binding\_scope} label, lifecycle metadata, and optional policy and
attribute fields. The wire envelope additionally carries a public-key-derived
\texttt{key\_id}. The deployer primary-signs a deterministic, recursively
key-sorted JSON serialization of the complete record with Ed25519.

The composed verifier resolves the deployer in a trusted local registry,
checks that the envelope \texttt{key\_id} matches that registry entry and has
not been revoked, verifies the primary signature and validity interval, and
recomputes the output hash over the exact bytes supplied by the consumer. It
also rejects duplicate nonces in a presented chain and, when a caller supplies
a persistent nonce store, a nonce previously consumed by that verifier.
These checks establish integrity and origin under the unified trust boundary
in Section~\ref{sec:threat_model}; they do not establish that the content is
safe or policy-compliant.

\subsection{Current Granularity Boundary}
\label{sec:granularity}

The schema supports the labels \texttt{per\_session}, \texttt{per\_turn},
\texttt{per\_tool\_call}, and \texttt{per\_n\_tokens}. They are signed so
that an application can state which bytes a record is intended to cover, but
the prototype does \emph{not} implement session identifiers, segment indices,
tool-call identifiers, or a streaming sequence verifier. Its executable
LangGraph evaluation therefore signs one deterministic output per graph node;
it is not evidence for streaming or token-window enforcement. A deployment
requiring those semantics must add and verify explicit session and sequence
fields before making claims about reordering or cross-session transplantation.

\subsection{Scope of the Evidence}

The signed record supports a post-hoc audit: an operator can identify the
deployer that released a byte sequence and correlate the record with its
delegation ancestry. It is compatible with information-flow control,
rule-based monitors, or semantic judges, but it does not replace them. If a
trusted runtime signs an attacker-directed output, the record faithfully
attests to that fact; preventing the output requires a separate control before
release.

\section{Ancestry Attestation (Layer~2)}
\label{sec:designs}

Layer~2 asks what evidence a verifier needs to reconstruct a dynamically
created delegation path. We start with familiar single-signer constructions as
comparative baselines, then add the property required for cross-deployer
authorization under child-key compromise. The primitives themselves are
standard; the design comparison makes explicit which deployer authorizes each
edge and what is lost when that key is compromised.

\subsection{Comparison Criteria}

Each node contains the issuer, instance and system identifiers, parent
reference, nonce, declared ancestry depth, expiry time, output commitment, and
chain-proof fields. We evaluate the
constructions against three requirements: (R1) bind a node to its parent so
that structural edits B1--B4 are detectable under uncompromised keys;
(R2) permit complete-path verification by a consumer without a public log; and
(R3) require parent consent when a child under a different deployer claims a
cross-deployer parent. R3 is the differentiator in the key-compromise setting.

\subsection{Baseline 1: Signed Linked List}
\label{sec:design_linkedlist}

The first baseline adds a signed parent commitment to an otherwise independent
per-node signature:
\[
  \texttt{prev\_hash} = \mathrm{SHA256}\!\left(\textsf{canonical\_json}(\text{parent})\right).
\]
The issuing deployer signs its complete record, including
\texttt{prev\_hash}. A verifier traverses the presented path, checks each
issuer signature and recomputes the commitment. This conventional hash-chain
construction satisfies R1 and R2: removing a middle node, transplanting a
node, or modifying an ancestor breaks a descendant commitment when all relevant
signing keys remain secret. Its verification cost is one hash and one signature
check per hop.

The baseline does not satisfy R3. A child deployer can sign a new record that
points to any publicly visible parent record. Its \texttt{prev\_hash} is
well-formed, but the parent never approved the new relationship. This is not a
hash-chain failure: it is a missing authorization statement from the parent.

\subsection{Baseline 2: Merkle-Chain Variant}
\label{sec:design_merkle}

The second baseline replaces the per-hop parent commitment with a Merkle
commitment~\cite{merkle1980protocols} over the current ancestry prefix:
\[
  \texttt{merkle\_root} = \textsf{MerkleRoot}(v_0,\ldots,v_k),
\]
where leaves are canonical node representations. This is a standard way to
obtain logarithmic-size inclusion proofs when a consumer verifies a selected
ancestor rather than the entire path. For the full-path verification required
by R2, it adds another consistency check but not a different authorization
principal. A compromised child can recompute the affected root and sign the
new node, so the construction also fails R3.

We include the Merkle variant because it is a natural alternative often
suggested for provenance. It may be useful when a deployment genuinely needs
partial-path proofs. That optimization is outside this evaluation, whose
consumer verifies the complete ancestry to detect omission and splice attacks.

\subsection{Cross-Deployer Design: Co-Signed DAG}
\label{sec:design_signeddag}

The \signeddag construction adds an authorization statement for each
cross-deployer edge. For an edge from parent node $v_A$ at deployer $D_A$ to
child node $v_B$ at deployer $D_B$, both deployers sign the same
domain-separated canonical edge manifest:
\begin{multline*}
  \bar v_B = v_B[\texttt{chain\_proof.signatures} \leftarrow \varnothing],\\
  \tau = \textsf{Canon}(\texttt{domain},
        \mathrm{SHA256}(\textsf{Canon}(v_A)),\bar v_B),\\
  \sigma_{AB}=\bigl(\mathrm{Sig}_{sk_A}(\tau),\mathrm{Sig}_{sk_B}(\tau)\bigr).
\end{multline*}
Thus, $\tau$ binds the parent receipt digest and every child field, including
the child deployer, output hash, expiry, nonce, lifecycle record, proof design,
and co-signer policy. Only the signature list being created is blanked to avoid
a signature cycle; the child's primary signature then seals the completed
record. The explicit domain tag prevents a valid edge signature from being
reinterpreted by another protocol.
The child signature attests that $D_B$ accepted the delegation; the parent
signature attests that $D_A$ authorized this specific edge. The verifier
requires both signatures for cross-deployer edges and checks parent references
and declared depth for the rest of the path. Same-deployer edges need only the
local deployer's signature.

\begin{figure}[!t]
    \centering
    \includegraphics[width=0.98\columnwidth]{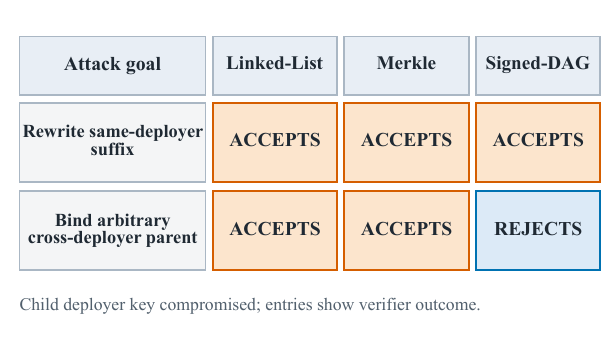}
    \caption{Child-key-compromise outcomes. All three designs accept a
    same-deployer suffix rewrite, whereas only the co-signed DAG rejects an
    arbitrary cross-deployer parent binding because it requires the
    uncompromised parent's co-signature.}
    \label{fig:key_compromise}
\end{figure}

The property comes with an operational trade-off. The parent deployer must be
online when the edge is created, and a refused or timed-out counter-signature
causes the delegation to fail closed. A deployment can use the linked
representation for same-deployer hops and queue a cross-deployer request. The
in-process harness centralizes co-signer access for reproducible tests; the
process, container, and AWS experiments separately exercise remote signer RPC.

\subsection{Path Completeness and Lifecycle Checks}
\label{sec:lifecycle}

Signature and parent checks alone do not ensure that the consumer received a
complete path. Each record therefore contains a signed \texttt{ancestry\_depth};
the verifier rejects a presentation whose number of records differs from the
declared depth. Nonce checks prevent reuse of a record at the same verifier.

The implementation represents lifecycle events
(\texttt{new}, \texttt{reload}, \texttt{branch}, \texttt{composite}, and
\texttt{fine\_tune}). Their semantic invariants are separate from the
cryptographic predicates. The current evaluation does not claim a complete
lifecycle-semantics audit; a deployment must validate those invariants against
its own runtime event model.

\section{Evaluation}
\label{sec:results}

We evaluate three questions. First, what protection does each ancestry
construction provide, including under a compromised deployer key? Second, does
the implementation enforce the specified structural predicates on fixed and
adaptive inputs? Third, what is the cost of issuing and verifying the evidence?

\begin{figure*}[!t]
    \centering
    \includegraphics[width=0.80\textwidth]{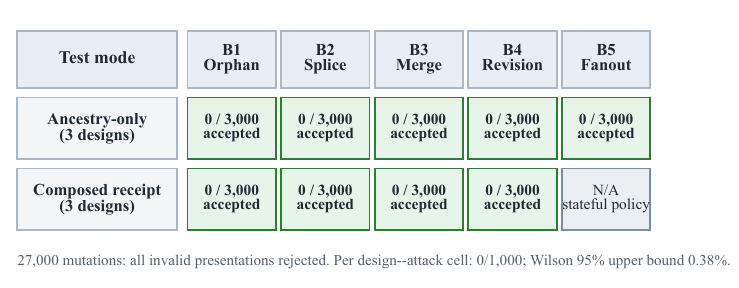}
    \caption{Fixed-mutation conformance coverage. Across 27{,}000 mutations, no
    invalid presentation was accepted. The composed-receipt matrix repeats B1--B4;
    B5 requires stateful fanout verification and is exercised only in ancestry-only
    mode. Each check aggregates three designs and 1{,}000 trials per design.}
    \label{fig:fixed_mutation_coverage}
\end{figure*}

\subsection{Methodology and Platform}\label{sec:setup}

The prototype is implemented in Python~3.11 with Ed25519 from
\texttt{cryptography}~$\geq$~42 and a deterministic local JSON serializer
(not a full RFC~8785/JCS implementation). Unless otherwise noted, experiments
ran on a MacBook Pro with an Apple M1 Pro CPU (six performance and two
efficiency cores), 16~GB unified memory, and macOS~15.7.4. The separate
E11--AWS run used three Amazon Linux~2023 EC2 hosts in distinct availability
zones; its configuration-specific timings are reported separately below.

For the current revision, we reran a fixed adversary matrix with 1{,}000
trials for each design--attack pair. The B1--B4 matrix has an ancestry-only
mode and a composed-receipt mode; the latter also checks each record's primary
signature, \texttt{key\_id}, expiry, output hash, and nonce predicate.
Unexpected verifier exceptions abort the experiment rather than being counted
as rejections. Table~\ref{tab:model_coverage} distinguishes these current,
deterministic checks from historical LLM-guided exploratory campaigns.

\begin{table}[t]
\centering
\caption{Model-to-harness coverage. A zero-success structural run checks
implementation conformance, not model strength or primitive security.}
\label{tab:model_coverage}
\scriptsize
\setlength{\tabcolsep}{2pt}
\renewcommand{\arraystretch}{1.08}
\begin{tabular}{@{}>{\raggedright\arraybackslash}p{1.25cm}>{\raggedright\arraybackslash}p{1.15cm}>{\raggedright\arraybackslash}p{1.0cm}>{\raggedright\arraybackslash}p{3.0cm}@{}}
\toprule
Harness & Driver & Status & Coverage / outcome \\
\midrule
Ancestry-only & Python & Current & 3 designs $\times$ B1--B5 $\times$ 1{,}000; 0 accepted per cell \\
Composed receipt & Python & Current & 3 designs $\times$ B1--B4 $\times$ 1{,}000; 0 accepted per cell \\
Structural & LLMs & Historical & Linked-list-focused exploratory campaign; not rerun after current edge revision \\
Framework & LangGraph 1.2.6 & Current & 1{,}000 deterministic runs; composed verifier accepts all valid chains \\
Live workflow & MCP + A2A + Ollama & Current & 30 tasks/condition; actual tool and LLM bytes; S-DAG verifies every valid task \\
Process S-DAG & Python \texttt{spawn} + Pipe & Current & 1{,}000 valid cross-deployer paths; stale-key reject and fresh-key rotation check \\
Container S-DAG & Docker bridge + TCP & Current & 1{,}000 valid paths; replay, registry, revocation, and refusal faults fail closed \\
Cloud S-DAG & AWS EC2 + SSM & Current & 1{,}000 valid paths across three zones; the same four faults fail closed \\
\bottomrule
\end{tabular}
\end{table}

The current matrix is reproducible and measures conformance on a finite,
hand-specified attack set; its confidence interval is not a cryptographic or
open-world security guarantee. The older LLM-guided aggregate files are
retained as exploratory context but are not presented as validation of the
revised complete-verifier or signed-edge implementation.

\subsection{Structural Security Comparison}
\label{sec:results_baselines}

We first isolate the effect of ancestry evidence. \nochain carries no
inter-node relationship, \hashonly adds an unsigned parent hash, and
\naivesig signs only the node's own identity. The remaining three designs are
the comparative constructions in Section~\ref{sec:designs}. Table~\ref{tab:baselines}
summarizes the expected verifier outcome for the attack taxonomy under
uncompromised keys. In particular, independent node signatures cannot detect a
middle-node splice (B2), because every remaining node still verifies on its
own. A signed parent commitment is the minimal additional evidence for that
case. B5 is policy-dependent. It is retained as a separate stateful
fanout-detector condition in the fixed matrix, rather than evidence for a
stateless cryptographic ancestry predicate.

\begin{table}[t]
\centering
\caption{Structural outcome under uncompromised keys. \checkmark\ = attacker
can construct an accepted chain; $-$ = verifier rejects.}
\label{tab:baselines}
\scriptsize
\setlength{\tabcolsep}{2pt}
\renewcommand{\arraystretch}{1.08}
\begin{tabular}{@{}lcccccc@{}}
\toprule
& \nochain & \hashonly & \naivesig & \linkedlist & \merkle & \signeddag \\
\midrule
\Bone & \checkmark & \checkmark & $-$ & $-$ & $-$ & $-$ \\
\Btwo & \checkmark & \checkmark & \checkmark & $-$ & $-$ & $-$ \\
\Bthree & \checkmark & \checkmark & $-$ & $-$ & $-$ & $-$ \\
\Bfour & \checkmark & \checkmark & $-$ & $-$ & $-$ & $-$ \\
\midrule
\Bfive & \checkmark & \checkmark & \checkmark & $-^\dagger$ & $-^\dagger$ & $-^\dagger$ \\
\bottomrule
\multicolumn{7}{@{}l}{\scriptsize $^\dagger$Requires a signed \texttt{max\_children} policy and verifier state.}
\end{tabular}
\end{table}

\subsubsection{Child-key compromise}\label{sec:results_keycompromise}

The distinguishing case is a compromised child deployer $D_B$ attempting to
claim an arbitrary parent $D_A$. In the linked and Merkle variants, $D_B$ can
compute the public parent commitment and sign the new child record; neither
format asks $D_A$ to approve that edge. In the co-signed DAG, the edge
transcript also needs $D_A$'s signature. This is the central security
difference, rather than a claim that any of the underlying primitives is new.

\begin{definition}[Edge-authorization game]
\label{def:kc_game}
The challenger generates deployer keypairs and gives the adversary one child
key $sk_i$. The adversary may obtain signatures for legitimate edge transcripts
from uncompromised deployers. It wins if it creates an accepted cross-deployer
edge to an uncompromised parent that did not sign that transcript.
\end{definition}

\begin{theorem}[Cross-deployer edge authorization]\label{thm:kc}
For a \signeddag edge, under Ed25519 SUF-CMA security and SHA-256 collision
resistance, the probability of winning the edge-authorization game is
negligible in the security parameter. For \linkedlist and \merkle, the
probability is one when the compromised party is the child endpoint, because
the parent does not sign the edge.
\end{theorem}

\begin{proof}
An accepted \signeddag edge contains the uncompromised parent's signature over
the complete edge transcript. A new unauthorized transcript therefore yields a
signature forgery, except with negligible probability. The two single-signer
formats require only the child's valid signature and a correctly computed,
public parent commitment, so the compromised child can create the record
directly. \qed
\end{proof}

Table~\ref{tab:key_compromise} identifies the limits of the result. It proves
edge creation only; it does not provide policy-semantic authorization,
revocation freshness, or resistance when both endpoints collude.

\begin{table}[t]
\centering
\caption{Key-compromise analysis.}
\label{tab:key_compromise}
\scriptsize
\setlength{\tabcolsep}{2pt}
\renewcommand{\arraystretch}{1.12}
\begin{tabular}{@{}>{\raggedright\arraybackslash}p{1.55cm}>{\raggedright\arraybackslash}p{1.9cm}cc@{}}
\toprule
Compromised party & Attack goal & LL/Merkle & S-DAG \\
\midrule
Child $D_B$ & Bind to arbitrary parent & accepts & rejects \\
Parent $D_A$ & Rewrite parent after child reference & rejects$^\dagger$ & rejects \\
Same deployer & Rewrite entire suffix & accepts & accepts \\
Both endpoints & Forge cross-deployer edge & accepts & accepts \\
\bottomrule
\end{tabular}\\[2pt]
{\scriptsize $^\dagger$The child's signed parent commitment anchors the original parent.}
\end{table}

\subsection{Implementation-Conformance Validation}
\label{sec:results_fixed}

The current ancestry-only matrix covers three designs, five attack classes
(B1--B5), and 1{,}000 trials per cell; all 15 cells rejected every trial
(Wilson 95\% upper bound 0.38\% per cell). A separate composed-receipt matrix
repeats B1--B4 for all three designs, again with zero accepts in all 12 cells.
B1--B4 exercise malformed chain presentations; B5 exercises the separate local
fanout detector. Figure~\ref{fig:fixed_mutation_coverage} summarizes test-mode
coverage and the common zero-acceptance outcome. This result is not an empirical security bound: under
the assumptions of Section~\ref{sec:threat_model}, a successful edit would
indicate an omitted or incorrect verifier predicate. Its value is
implementation validation.

Archived LLM-guided and lifecycle-mutation aggregates lack prompts, variants,
or revision metadata. We retain them as exploratory context only; the current
deterministic matrices are the evidence for the revised implementation.

\subsection{Performance and Framework-Integration Overhead}\label{sec:results_perf}

At ancestry depth five over 1{,}000 trials, mean ancestry-only checks took
24.3\,$\mu$s for \linkedlist, 72.7\,$\mu$s for \merkle, and 499.2\,$\mu$s for
\signeddag. They exclude primary receipt checks and are not complete-verifier
costs. The Signed-DAG value includes verification of the revised full-child
edge manifest; we do not report a stale serialized-evidence size from the
previous edge format.

We also integrated the \linkedlist signing hook into a compiled LangGraph~1.2.6
\texttt{StateGraph} with three executed nodes (Planner $\rightarrow$ Researcher
$\rightarrow$ Writer), where the writer runs under a second deployer. Over
1{,}000 executions after 100 warm-up runs, every three-node chain passed the
composed verifier (primary signature, receipt hash, expiry, nonce uniqueness,
and linked-list ancestry). Signing averaged 0.171\,ms per node (0.214\,ms at
p95), or 0.512\,ms per task (0.663\,ms at p95); receipt--chain verification
averaged 0.832\,ms (1.013\,ms at p95). Nodes use deterministic outputs and no
LLM client, so these are local protocol-overhead measurements, not model-serving
latency; the experiment also omits \signeddag's online co-signature handshake.

\paragraph{Live MCP/A2A/LLM workflow (E12)}
We next execute the protocol path that the local LangGraph benchmark omits.
Three loopback-only processes represent separate deployer trust domains: a
parent signer, a child A2A JSON-RPC agent, and an MCP Python-SDK
Streamable-HTTP tool server. For each task, the parent discovers the child's
Agent Card, invokes A2A \texttt{message/send}, and the child calls the MCP
\texttt{get\_evidence} tool before a local Ollama \texttt{llama3.1:latest}
model (temperature zero) produces the released text. The child signer then
returns either no evidence, a child-primary-signed \linkedlist receipt, or a
\signeddag receipt after an online parent co-signature over the exact
provisional child transcript. The coordinator reconstructs a public-only
registry and uses the complete verifier with a fresh file-backed nonce store
for each presentation.

After ten warm-up tasks per condition, we ran 30 tasks per condition with a
rotated condition order. Table~\ref{tab:e12_live} reports the entire path,
including discovery, A2A, MCP, local LLM, issuance, and complete verification;
all conditions use the same fixed tool evidence and prompt. Every attested task
passed complete verification and no task aborted. The Signed-DAG mean is
42.3\,ms (5.5\%) above the no-evidence path; its 11.6\,ms child issuance and
16.7\,ms complete verification are small relative to its 681.3\,ms mean LLM
generation. It carries 2{,}096 bytes of evidence, 228 bytes more than the
linked-list condition. This is a task-level measurement, not a claim that
local-model output is deterministic or that the LLM is semantically safe.

\begin{table}[t]
\centering
\caption{E12 live workflow: 30 tasks per condition after ten warm-ups. E2E
includes A2A, MCP, and local LLM; issuance and verification apply to attested
conditions. The controlled attack is not an attack-success-rate estimate.}
\label{tab:e12_live}
\scriptsize
\setlength{\tabcolsep}{2.3pt}
\renewcommand{\arraystretch}{1.08}
\begin{tabular}{@{}>{\raggedright\arraybackslash}p{1.25cm}ccc>{\raggedright\arraybackslash}p{1.35cm}@{}}
\toprule
Condition & Valid & E2E mean/p95 (ms) & Evidence (B) & Child-key-only claim \\
\midrule
No evidence & 30/30 & 770.8 / 846.0 & 0 & Unauditable \\
\linkedlist & 30/30 & 826.2 / 941.1 & 1{,}868 & Accepts \\
\signeddag & 30/30 & 813.1 / 927.7 & 2{,}096 & Rejects \\
\bottomrule
\end{tabular}
\end{table}

\paragraph{Cross-deployer deployment evidence}
E10 (\texttt{spawn}/Pipe) and containerized E11 each verified 1{,}000
two-signer paths beyond the in-process harness; their timing and fault traces
are retained in the artifact. We next report the separately configured AWS run.

\paragraph{Key deployment result: E11--AWS}
We ran E11 on three \texttt{t3.small} EC2 hosts in distinct availability zones:
parent signer, child signer, and public-only coordinator. Digest-pinned
containers retained signer keys; SSM managed hosts, and authenticated TCP used
private addresses with signer ingress limited to the coordinator security group.
After 100 warm-ups and 1{,}000 trials, every path verified and none aborted.
Issuance averaged 3.651\,ms (4.384\,ms at p95), and complete verification
5.015\,ms (6.642\,ms at p95). Durable replay, an unregistered parent, and a
revoked child key were rejected; parent refusal aborted issuance. These
configuration-specific measurements exclude TLS/mTLS, KMS signing, model-client,
and production authorization-service latency.

\section{Deployment Implications and Limitations}\label{sec:discussion}

The design choice depends on the trust relationship. A signed linked list had
the lowest measured local ancestry-verification cost within one deployer trust
domain. A Merkle commitment can support partial-path proofs, which we do not
evaluate. A co-signed DAG is appropriate when a deployment requires explicit
parent authorization for cross-deployer edges and can tolerate an online
co-signature handshake.

The LangGraph measurement uses deterministic outputs, while E12 adds a local
MCP/A2A workflow with actual local-LLM output. E12's loopback services are
separate deployer trust domains, not independent organizations or AWS accounts;
it does not measure public-network latency, TLS/mTLS, OAuth, KMS, or production
authorization-service latency. Archived LLM-generated aggregates are not current
security evidence. A compromised runtime, open registry, cross-verifier replay
prevention, production key-rotation protocol, endpoint collusion,
streaming-order semantics, LLM quality, and application-policy enforcement
remain outside the model.

\section{Conclusion}\label{sec:conclusion}

We have presented deployer-side evidence for outputs and delegation ancestry in
multi-agent systems. The core result is that a child-signed parent commitment
does not authorize a cross-deployer edge after child-key compromise, whereas a
parent co-signature over the complete edge manifest does. The prototype confirms
the specified predicates and verifies the co-signed-DAG path in process,
container, three-zone AWS, and live local MCP/A2A/LLM workflows. It supports
post-incident audit and attribution, not prompt-injection prevention. The
anonymous artifact package, including current sources, raw results, manifests,
and regression commands, is available at
\url{https://anonymous.4open.science/r/cscloud-artifacts-6342/}.

\bibliographystyle{IEEEtran}
\bibliography{references}

\end{document}